\documentclass[a4paper,10pt,twoside,reqno,nonamelimits]{article}
\usepackage{a4wide}
\usepackage[colorlinks=true,urlcolor=blue,linkcolor=blue,citecolor=blue]{hyperref}
\usepackage{bookman}         

\usepackage{amsmath}
\usepackage{amssymb}
\usepackage{amsfonts}
\usepackage{amscd}
\usepackage{amsthm}
\usepackage{amsxtra}
\usepackage{mathabx}
\usepackage{mathrsfs}
\usepackage{nicefrac}
\usepackage{graphicx}
\usepackage{xcolor}
\usepackage{tikz}
\usepackage{wrapfig}
\usepackage{enumitem}
\setlist[enumerate]{leftmargin=2em,itemindent=0em, labelindent=0pt,labelwidth=1.5em,labelsep=.5em, align=left, noitemsep}
\newlist{txtenum}{enumerate}{1}
\setlist[txtenum]{leftmargin=0em,itemindent=1.5em, labelindent=0pt,labelwidth=1em,labelsep=.5em, align=left}

\theoremstyle{plain}
\newtheorem{theorem}{Theorem}
\newtheorem*{theorem*}{Theorem}

\newtheorem*{proposition*}{Proposition}

\newtheorem*{corollary*}{Corollary}

\newtheorem{lemma}[theorem]{Lemma}
\newtheorem*{lemma*}{Lemma}

\newtheorem*{observation*}{Observation}

\newtheorem*{conjecture*}{Conjecture}

\newtheorem*{question*}{Question}

\newtheorem*{questions*}{Questions}

\newtheorem*{problem*}{Problem}

\newtheorem*{problems*}{Problems}

\newtheorem*{openproblem*}{Open Problem}

\theoremstyle{definition}

\newtheorem*{definition*}{Definition}
\newtheorem{example}[theorem]{Example}
\newtheorem*{example*}{Example}

\newtheorem*{exercise*}{Exercise}

\newtheorem{remark}[theorem]{Remark}
\newtheorem*{remark*}{Remark}

\newtheorem*{remarks*}{Remarks}
\newtheorem{fact}[theorem]{Fact}

\theoremstyle{remark}
\newtheorem{claim}[theorem]{Claim}
\newtheorem*{claim*}{Claim}

\newcommand{\subclass}[1]{}

\newcommand{\enumTi}[1]{\renewcommand{\theenumi}{#1}}

\newcommand{\alphenumi}{\enumTi{\alph{enumi}}}

\newcommand{\romenumi}{\enumTi{\roman{enumi}}}

\newlength{\hspaceforlengthglumpf}

\renewcommand{\em}{\sl}

\newcommand{\comment}[1]{\text{\footnotesize[#1]}}

\newcommand{\One}{\mathbf{1}}

\newcommand{\lt}{\left}
\newcommand{\rt}{\right}

\newcommand{\abs}[1]{{\lt\lvert{#1}\rt\rvert}}

\newcommand{\nfrac}[2]{{\nicefrac{#1}{#2}}}

\newcommand{\ZZ}{\mathbb{Z}}

\DeclareMathOperator*{\Prb}{\mathbb{P}}
\DeclareMathOperator*{\Exp}{\mathbb{E}}

\DeclareMathOperator{\IndicatorOp}{\mathbf{1}}
\newcommand{\Ind}{\IndicatorOp}

\newcommand{\eps}{\varepsilon}

\newlength{\algotabbingwidth}
\newcommand{\mypar}{\par\smallskip\noindent}

\newcommand{\todo}[1]{}
\newcommand{\TODO}[1]{\Yikes}
\newcommand{\TODOtxt}[1]{\Yikes}
\newcommand{\TODOpar}[1]{\Yikes}
\newcommand{\status}[1]{}

\newcommand{\sq}{\square}
\newcommand{\I}[1]{\One_{#1}}

\newcommand{\CmnE}{E^\cap}
\newcommand{\CmnEa}{E^{\cap\ast}}          \newcommand{\Sa}{S^{\ast}}
\newcommand{\RstE}{E^r}                    \newcommand{\Rst}{R}

\begin{document}
\title{On the Graph of the Pedigree Polytope}
\author{Abdullah Makkeh, Mozhgan Pourmoradnasseri, Dirk Oliver Theis\thanks{Supported by the Estonian Research Council, ETAG (\textit{Eesti Teadusagentuur}), through PUT Exploratory Grant \#620, and by the European Regional Development Fund through the Estonian Center of Excellence in Computer Science, EXCS.}\\[1ex]
  \small Institute of Computer Science {\tiny of the } University of Tartu\\
  \small \"Ulikooli 17, 51014 Tartu, Estonia\\
  \small \texttt{\{mozhgan,dotheis\}@ut.ee}%
}
\maketitle

\begin{abstract}
  Pedigree polytopes are extensions of the classical Symmetric Traveling Salesman Problem polytopes whose graphs (1-skeletons) contain the TSP polytope graphs as spanning subgraphs.

  While deciding adjacency of vertices in TSP polytopes is coNP-complete, Arthanari has given a combinatorial (polynomially decidable) characterization of adjacency in Pedigree polytopes.  Based on this characterization, we study the graphs of Pedigree polytopes asymptotically, for large numbers of cities.

  Unlike TSP polytope graphs, which are vertex transitive, Pedigree graphs are not even regular.  Using an ``adjacency game'' to handle Arthanari's intricate inductive characterization of adjacency, we prove that the minimum degree is asymptotically equal to the number of vertices, i.e., the graph is ``asymptotically almost complete''.

  \medskip%
\end{abstract}

\section{Introduction}\label{sec:intro}
The graph (1-skeleton) of a polytope has as its vertices (edges) the vertices (edges) of the polytope.  The most venerable result on graphs on polytopes: Steinitz's Theorem states that 3-connected planar graphs are precisely the graphs of 3-dimensional polytopes.

Properties of graphs of polytopes of higher dimension are of interest not only in the combinatorial study of polytopes, but also in Combinatorial Optimization, and Theoretical Computer Science.

For example, the famous Hirsch conjecture in the combinatorial study of polytopes, settled by Santos~\cite{Santos:hirsch:2012}, concerned the diameter of graphs of polytopes.

In Combinatorial Optimization, the study of the graphs of polytopes associated with combinatorial optimization problems was initially motivated by the search for algorithms for these problems.

In Theoretical Computer Science, the theorem by Papadimitriou~\cite{Papadimitriou:adj-tsp:78} that Non-Adjacency of vertices of (Symmetric) Traveling Salesman Problem (TSP) polytopes is NP-complete, gave rise to similar results about other families of polytopes (cf.\ \cite{Aguilera:adj-setcover-polyh:2014,Maksimenko:adj-np:2014} and the references therein, for recent examples).

There have been particularly many attempts to understand the graph of TSP polytopes, and, where this turned out to be infeasible, of TSP-related polytopes (e.g., \cite{Sierksma:DAM:00}; cf.\ \cite{Groetschel-Padberg:polyh-th:1985,NaddefRinaldi93,Theis:DiscreteOpt:2014}).  The presence of long cycles has been studied (\cite{Sierksma:skeleton:1993}, see also \cite{Naddef-Pulleyblank:JCTB:84,Naddef:JCTB:84}), as has the graph density / vertex degrees (e.g., \cite{Sarangarajan:tspskel-density-LB:1997}, see also \cite{Kaibel-Remshagen:03,Kaibel:lowdim-faces-rndptp:2004}).

The original motivation for the research in this paper was a 1985 conjecture by Gr\"otschel and Padberg~\cite{Groetschel-Padberg:polyh-th:1985} --- well-known in polyhedral combinatorial optimization (Problem \# 36 on Schrijver's list~\cite{Schrijver:Book:03}) --- stating that the graph of TSP polytopes has diameter~2.  Already in~\cite{Groetschel-Padberg:polyh-th:1985}, Gr\"otschel and Padberg extend the question for the diameter to a family of TSP-related polytopes which seemed easier to understand at the time.

\begin{figure}[htbp]%
  \centering%
  \scalebox{.4}{\includegraphics{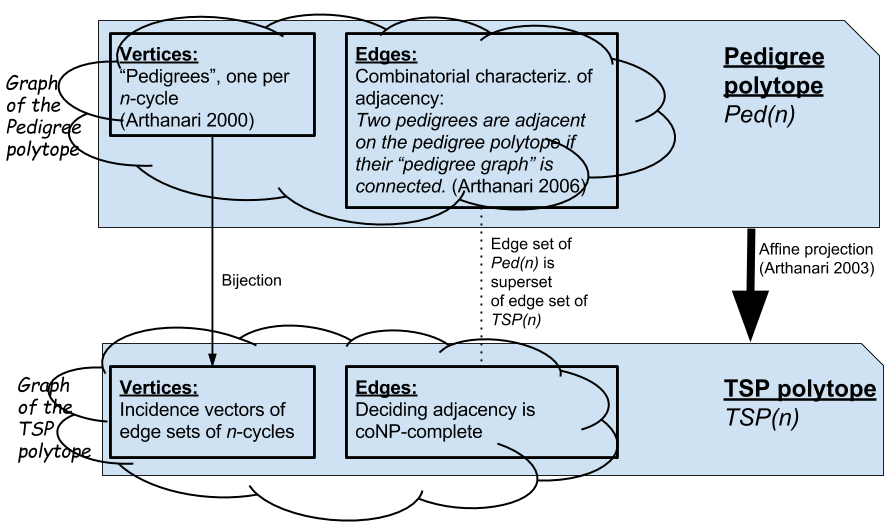}}%
  \caption{Polytopes and graphs}\label{fig:polytopes-graphs-diagram}%
\end{figure}

A more recent family of TSP-related polytopes are the \textit{Pedigree polytopes} of Arthanari~\cite{Arthanari:ped-one:2000}.  For this family of polytopes, adjacency of vertices can be decided in polynomial time~\cite{Arthanari:DiscreteMath:06}.  Moreover, the graphs of the TSP polytopes are spanning subgraphs of the graphs of the Pedigree polytopes~\cite{Arthanari:DiscreteOpt:13}.  This follows as the Pedigree polytope for~$n$ cities is an \textit{extension,} without ``hidden''~\cite{Pashkovich:hidden:15} vertices, of the TSP polytope~\cite{Fiorini-Kaibel-Pashkovich-Theis:CombLB:13}.  The vertex set of the TSP polytope for~$n$ cities is in a natural bijection with the set of all cycles with node set $[n] := \{1,\dots,n\}$; the same is true for the vertex set of the Pedigree polytope for~$n$ cities.  See Fig.~\ref{fig:polytopes-graphs-diagram}.

The main result of our paper, is the following fact about graphs of Pedigree polytopes.  Recall that the number of vertices of either the TSP polytope or the Pedigree polytope for~$n$ cities is the number of $n$-cycles, which is $(n-1)!/2$.

\begin{theorem}\label{thm:main-polytope}
  The minimum degree of a vertex on the Pedigree polytope for~$n$ cities is $(1-o(1))\cdot(n-1)!/2$ (for $n\to\infty$).
\end{theorem}

In particular, the density graph of Pedigree polytopes is asymptotically equal to~1.  Note, though, that while for TSP polytopes, these two statements are equivalent, this is not the case of Pedigree polytopes.  The reason is that Pedigree polytopes are not as ``symmetric'' as TSP polytopes: For every two vertices $u,v$ of the TSP polytope for~$n$ cities, there is an affine automorphism of the polytope mapping $u$ to~$v$.  (Similar statements are true for monotone-TSP and graphical-TSP polytopes.)  This is not true for Pedigree polytopes: Arthanari's construction removes the symmetry to a large extent.

Numerical simulations show that, even for relatively large $n$ (say, $\approx 100$), the graph of the Pedigree polytope is not complete.  We have made no attempt, however, to find a non-trivial upper bound for the minimum degree.

\bigskip\noindent%
We now give a non-technical description of the proof of Theorem~\ref{thm:main-polytope}.
\paragraph{The adjacency game.}
Alice and Bob play a game on a graph $\Gamma$ (which both of them see): Alice picks a vertex, then Bob picks a vertex; if the two vertices are adjacent, Bob wins, otherwise Alice wins.  The game is silly, Bob always wins (unless there is an isolated vertex).  Let us blindfold Bob --- he can no longer see the graph or Alice's move, so the best he can do is pick a random vertex.  Alice will win the game with probability $\mbox{mindeg}(\Gamma)/ \abs{V(\Gamma)}$.
The game is still trivial, but note that Alice's chance of winning is linked to the minimum degree.

The game becomes interesting if the graph evolves over time.  At each time~$n$, each vertex of the current graph $\Gamma_{n-1}$ will be replaced by a set of child vertices.  Complex rules (known to Alice\footnote{Bob doesn't need them since he plays randomly.}) govern whether or not a particular child vertex inherits the adjacency relation to a particular child of a neighbor of his parent.  Also at each time, both Alice and Bob have to update their choices for vertices: Alice can pick any child of her current vertex; Bob, being blindfolded, picks a child of his current vertex uniformly at random.

Now, analyzing the game as a random process (asymptotically for $n\to \infty$) becomes a non-trivial task, which reveals the minimum degree of $\Gamma_n$, as $n\to \infty$.  (Every vertex of $\Gamma_n$ must have the same number of children for Bob's random decisions to form a uniformly random vertex.)

We haven't made clear what exactly Alice's goal is, or how payout takes place, but these questions will fall in place naturally in our situation (for $n\to\infty$, with probability tending to~1, Alice will lose in every round).

Our graph $\Gamma_n$ is the graph of the pedigree polytope for~$n$ cities, whose vertices are ``pedigrees''.  Stricly speaking, pedigrees are combinatorial-geometric objects defined by Arthanari, but to reduce technical overhead, we will work with cycles \textsl{only.}  So, at time~$n$, Alice and Bob each holds a cycle with node set~$[n-1]$.  A child-vertex is formed from a parent-vertex by inserting the new node~$n$ into the cycle.  Alice picking one of these children amounts to inserting the new node~$n$ into her cycle; Bob inserts the new node~$n$ at a uniformly random position into his cycle.

As for the complex rules governing inheritance of adjacency: these are given by Arthanari's characterization of adjacency on Pedigree polytopes~\cite{Arthanari:DiscreteMath:06}, which are best understood as a process, updating a combinatorial structure at each time.

\paragraph{Pedigrees and how they are adjancent.}
Arthanari's beautiful idea of a pedigree is that of a cycle ``evolving'' over time: Starting from the unique cycle with node set $\{1,2,3\}$ at time~$3$, at time~$n\ge 4$, the node~$n$ is added to the cycle by subdividing one of its cycle-edges.  We say that $n$ is \textit{inserted into} that cycle-edge.

Arthanari's combinatorial condition for adjancency on the Pedigree polytope has process character, too, with a combinatorial structure, the \textit{pedigree graph}~$G$, evolving over time.  Suppose we have two evolving cycles.  Let us refer to $A$ as Alice's cycle, and to $B$ as Bob's cycle.  At time~$n$, Alice chooses a cycle-edge of her current cycle~$A$ (with node set $[n-1]$) and inserts her new node~$n$ into that cycle-edge to form her new cycle (with node set $[n]$).  Then Bob chooses a cycle-edge of his current cycle~$B$, and inserts his new node~$n$ into that cycle-edge to form his new cycle.

The pedigree graph~$G$ may also change at time~$n$.  The new pedigree graph is either equal to the current one, or arises from the current one by adding the new vertex\footnote{Trying to reduce confusion, in our terminology, the polytopes have \textit{vertices} which correspond to cycles; the cycles consist of \textit{nodes} and \textit{cycle-edges,} and the the pedigree graph of two cycles has \textit{vertices}.} $n$ with incident edges.  The choices of Alice and Bob determine: whether the new vertex is added or not; the number of edges incident to the new vertex~$n$; the end vertices of these edges.

Arthanari's combinatorial characterization of adjacency on the Pedigree polytope is now this.
\begin{theorem}[\cite{Arthanari:DiscreteMath:06}]\label{thm:arthanari}
  At all times $n\ge 4$, the two vertices of the Pedigree polytope for~$n$ cities corresponding to the (new) cycles $A$ and~$B$ with node set~$[n]$ are adjacent in the Pedigree polytope, if, and only if, the (new) graph~$G$ is connected.
\end{theorem}

Theorem~\ref{thm:main-polytope} states that, if~$B$ is a cycle chosen uniformly at random from all cycles on~$[n]$, then
\begin{equation*}
  \min_{A} \Prb\bigl( \text{$\{$ the pedigree graph is connected $\}$} \bigr) = 1-o(1),
\end{equation*}
where the minimum ranges over all cycles on~$[n]$.  Lower bounding this quantity amounts to studying the following adjacency game: Alice's goal is to make the graph~$G$ disconnected; whereas Bob makes uniformly random choices all the time.  We prove that Alice loses with probability $1-o(1)$.  To analyze the game, we study a kind of a Markov Decision Process with state space $\ZZ_+\times \ZZ_+$.  The states are pairs $(s,t)$, where~$s$ is the number of common cycle-edges in Alice's and Bob's cycles, and~$t$ is the number of connected components of the current pedigree graph.

\paragraph{In the next section,}
we will give rigorous statements corresponding to the hand-waving explanations above.  In Section~\ref{Sec:rnd-cyc-basics}, we discuss some basic facts about Bob playing randomly, and discuss the intuition of the proof of the main.  That section is followed by a more technical section containing the proofs of the basic properties of pedigree graphs, random or deterministic.  In Section~\ref{Sec:cnct-game}, we introduce the Markov-Decision-Problem-ish situation that Alice finds herself in.  The proof of the main theorem is completed in Section~\ref{Sec:mainproof}.  We conclude with a couple of questions for future research which we find compelling.

\section{Exact Statements of Definitions, Facts, and Results}\label{sec:exact-statements}
\subsection{Cycles, One Node at a Time}
Our cycles are undirected (so, e.g., there is only one cycle on 3 nodes).
For ease of notation, let us say that the \textit{positive direction} on a cycle with node set~$[n]$, $n\ge 3$, is the one in which, when starting from the node~1, the node~2 comes before the node~3; the other direction the \textit{negative direction.}
When referring to the $k$th cycle-edge of a cycle, we count the cycle-edges in the positive direction; the 1st one being the one incident on node~1.  E.g., in the unique cycle with node set $\{1,2,3\}$, the 1st cycle-edge is $\{1,2\}$, the 2nd is $\{2,3\}$, and the 3rd is $\{3,1\}$.

As mentioned in the introduction, Arthanari's Pedigree is a combinatorial object representing the ``evolution'' of a cycle ``over time'', and the combinatorial definition of adjacency of pedigrees makes use of that step-by-step development.  The set of Pedigrees is in bijecton with the set of cycles.  In our context (we do not have to associate points in space with Pedigrees), defining Pedigrees and then explaining the bijection with cycles is more cumbersome than necessary.  For convenience, we use the following more convenient definitions, which mirror the definition of Pedigrees, but they use cycles only.  Let us say that an \textit{infinite cycle}\footnote{The reason why we use this notion of ``infinite cycle'' is pure convenience.  It does not add complexity, but it makes many of statements and proofs less cumbersome.  Indeed, instead of an infinite cycle, it is ok to just use a cycle whose length~$M$ is longer than all the lengths occuring in the particular argument.  So instead of ``let $A$ be an infinite cycle, and consider $A_k$, $A_\ell$, $A_n$'' you have to say ``let $M$ be a large enough integer, $A_M$ a cycle of length~$M$, and $A_k$, $A_\ell$, $A_n$ sub-cycles of $A_M$''.  All the little arguments (e.g., Fact~\ref{fact:Bn-uniform} below) have to be done in the same way.} %
is a sequence $A = c_{\sq} \in \prod_{n=3}^\infty [n]$.  An infinite cycle~$A$ gives rise to an infinite sequence $A_{\sq}$ of finite cycles (in the usual graph theory sense), defined inductively as follows:
\begin{itemize}
\item $A_3$ is the unique cycle with node set $\{1,2,3\}$;
\item for $n\ge 3$, $A_n$ is the cycle with node set $[n]$ which arises from adding the node $n$ to $A_{n-1}$ by inserting it into (i.e., subdividing) the $c_{n-1}$th cycle-edge.
\end{itemize}
We think of $A_{\sq}$ as a cycle developing over time: At time~$n$, the node~$n$ is added.

We will need to access the neighbors of node~$n$ in~$A_n$, i.e., the ends of the cycle-edge into which~$n$ is inserted (i.e., which is subdivided) when~$n$ is added to~$A_{\sq}$.
We write $\nu_A^+(n)$ for the neighbor of~$n$ in~$A_n$ following~$n$ in the positive direction, and $\nu_A^-(n)$ for the neighbor of~$n$ in~$A_n$ following~$n$ in the negative direction.
The unordered pair $\nu_A(n) = \{ \nu_A^+(n), \nu_A^-(n) \}$ is the $c_{n-1}$th cycle-edge of $A_{n-1}$, the one into which~$n$ was inserted.

These definitions are for $n\ge 4$ but extend naturally for $n=1,2,3$: for $n=3$ we let $\nu_A^+(3)=1$, and $\nu_A^-(3)=2$; for $n=2$, we let $\nu_A^+(2)=\nu_A^-(2)=1$. The equation $\nu_A(n) = \{ \nu_A^+(n), \nu_A^-(n) \}$ holds for $n\ge 2$ (so $\abs{\nu_A(2)}=1$); for $n=1$ we have $\nu_A(1) := \emptyset$.

\begin{remark}[\it Finding $\nu(k)$ for ``old'' nodes~$k$]\label{rem:past-nu}
  It is readily verfied directly from the definition, that, for $k\ge 2$, $\nu^\pm_A(k)$ can be found as follows: start from node~$k$ and walk in positive direction.  The first node smaller than~$k$ which you encounter is $\nu^+_A(k)$.  Similarly, if you walk in negative direction starting from~$k$, the first node smaller than~$k$ which you hit, is $\nu^-_A(k)$.
\end{remark}

A pair of nodes $i,j$ split each cycle $A_{n}$, $n>i,j$ into two (open) segments ($i,j$ do not belong to either segment).  We say that the \textit{segment between $i$ and~$j$} is the one which does \textsl{not} contain the node $\min(\{1,2,3\}\setminus\{i,j\})$ (i.e., 1, unless $1\in\{i,j\}$, in that case, 2, unless $\{1,2\}=\{i,j\}$, in that case~3).  Note that this does not depend on the choice of $n>i,j$, which justifies to say \textit{``the segment of $A_{\sq}$ between $i$ and~$j$''.}

\begin{remark}[\it Testing/finding~$n$ with $\nu(n)=\{i,j\}$]\label{rem:check-if-pair-is-a-nu}
  Given a pair of nodes $\{i,j\}$ and $n'>i,j$, there exists an~$n\le n'$ with $\nu_A(n)=\{i,j\}$ if, and only if, the segment between $i$ and~$j$ on $A_{n'}$ is non empty and every node in it is larger than both $i$ and~$j$.
  In that case, the smallest node, $n$, in the segment between $i$ and~$j$ on $A_{\sq}$ is the one with $\nu(n)=\{i,j\}$.
\end{remark}

\subsection{The Pedigree Graph}
Two infinite cycles $A,B$ give rise to a sequence of graphs $G^{AB}_{\sq}$ which we call the \textit{pedigree graphs.}  We omit the superscripted $A,B$ when possible.  We speak of \textit{vertices} of the pedigree graphs (rather than nodes).  We do this to avoid confusion between the nodes of the cycles $A_{\sq}$,$B_{\sq}$ and the vertices of $G^{AB}_{\sq}$, because the vertex set of $G_n$ is a subset of $\{4,\dots,n\}$, and hence of the node set of $A_n$ and $B_n$.  So a node $k\in [n]$ may or may not be a vertex of $G_n$.

The pedigree graph $G_{n-1}$ is the subgraph of $G_n$ induced by the vertices in $[n-1]$.  In other words, $G_n$ is either equal to $G_{n-1}$ (if $n$ is not a vertex), or it arises from $G_{n-1}$ by adding the vertex~$n$ together with edges between~$n$ and vertices in $[n-1]$.

\begin{example}
  $G_1,G_2,G_3$ are graphs without vertices.  $G_4$ may be a graph without vertices, or it may consist of a single isolated vertex~$4$.  $G_5$ could be a graph without vertices; a graph with a single vertex~$5$; a graph with two isolated vertices $4$, $5$, or a graph with two vertices $4$, $5$, linked by an edge. Check figure \ref{fig.1} for possible $G_4$ and $G_5$.
\end{example}

According to Arthanari~\cite{Arthanari:DiscreteMath:06,Arthanari:DiscreteOpt:13} 
the condition for the existence of vertices is the following:

\begin{enumerate}[label=(\arabic*)]
\item A node~$n\in[n]$ is a vertex of $G_n$, iff $\nu_A(n) \ne \nu_B(n)$.
\end{enumerate}
There are several conditions for the presence of edges between the vertex~$n$ and earlier vertices.  To make it easier to distinguish these, we speak of edge ``types'' and give the edges implicit ``directions:'' from $A$ to~$B$ or from $B$ to~$A$.  Here are the conditions for edges from~$n$ to earlier vertices.
\begin{enumerate}[resume,label=(\arabic*)]
\item There is a \textit{type-1 edge ``from $A$ to~$B$''} between $n$ and $k\in [n-1]$, if $\nu_A(n) = \nu_B(k)$.  (Note that the condition implies that~$k$ is a vertex.)
\item There is a \textit{type-1 edge ``from $B$ to~$A$''} Ditto, with $A$ and~$B$ exchanged.
\item There is a \textit{type-2 edge ``from $A$ to~$B$''} between $n$ and~$\ell := \max\nu_A(n)$, unless $\displaystyle \nu_B( \ell ) \cap \nu_A(n) \neq \emptyset$.  In other words, suppose the node~$n$ was inserted into the cycle-edge $\{k,\ell\}$ in~$A$, with $k < \ell$.  Now look up the end-nodes of the cycle-edge $\nu_B(\ell)$ into which~$\ell$ was inserted when it was added to~$B$.  Unless~$k$ coincides with one of these end nodes, there is an edge between $n$ and~$\ell$.
\item \textit{Type-2 edge ``from $B$ to~$A$''} Ditto, with $A$ and~$B$ exchanged.
\end{enumerate}

Arthanari's theorem~\cite{Arthanari:DiscreteMath:06} (Theorem~\ref{thm:arthanari}) states that, if $n\ge 4$, and $A_n$, $B_n$ are two cycles with node set $[n]$, then the two vertices of the Pedigree polytope (for~$n$ cities) corresponding to $A_n$ and $B_n$ are adjacent, if, and only if, $G^{AB}_n$ is connected.

We will always think of~$A$ as ``Alice's cycle'' and $B$ as ``Bob's cycle''.

\begin{example}
  Going through an example will help understand the definition of a pedigree graph.  Figure~\ref{fig.1} shows two cycles $A$ and~$B$ evolving over time $n=3,\dots,10$, together with the evolving pedigree graph $G^{AB}_{\sq}$.

  \begin{figure}[hbp]
    \begin{center}
      \begin{tikzpicture}
        \draw (0,-7) -- (0,13); \draw (0,-7) -- (15,-7); \draw (0,13) -- (15,13); \draw (0,12) -- (15, 12); \draw (15,-7) -- (15,13);
        \draw (1,-7) -- (1,13); \draw (6,-7) -- (6,13); \draw (11,-7) -- (11,13);
        \draw (0,-3.5) -- (15,-3.5); \draw (0,0) -- (15,0); \draw (0,3) -- (15,3); \draw (0,5.5) -- (15,5.5); \draw (0,8) -- (15,8); \draw (0,10) -- (15,10); \draw (0,11) -- (15,11);
        \draw (0.5,12.5) node {$n$}; \draw (3.5,12.5) node {$A_n$}; \draw (8.5,12.5) node {$B_n$}; \draw (13,12.5) node {$G_n^{AB}$};
        \draw (0.5,-5.25) node {10}; \draw (0.5,-1.75) node {9}; \draw (0.5,1.5) node {8}; \draw (0.5,4.25) node {7}; \draw (0.5,6.75) node {6}; \draw (0.5,9) node {5}; \draw (0.5,10.5) node {4}; \draw (0.5,11.5) node {3};
        \draw (3.5,11.5) circle (0.4); \draw (3.5,10.5) circle (0.4); \draw (3.5,9) circle (0.9); \draw (3.5,6.75) circle (1.1); \draw (3.5,4.25) circle (1.1); \draw (3.5,1.5) circle (1.35); \draw (3.5,-1.75) circle (1.6); \draw (3.5,-5.25) circle (1.6);
        \draw (8.5,11.5) circle (0.4); \draw (8.5,10.5) circle (0.4); \draw (8.5,9) circle (0.9); \draw (8.5,6.75) circle (1.1); \draw (8.5,4.25) circle (1.1); \draw (8.5,1.5) circle (1.35); \draw (8.5,-1.75) circle (1.6); \draw (8.5,-5.25) circle (1.6);
        \shade[ball color=green] (3.2, 11.8) circle (2pt); \shade[ball color=green] (3.2, 11.2) circle (2pt); \shade[ball color=green] (3.9, 11.5) circle (2pt); \draw (3,11.8) node {\small 1}; \draw (3,11.2) node {\small 2}; \draw (4.1,11.5) node {\small 3};
        \shade[ball color=green] (3.2, 10.8) circle (2pt); \shade[ball color=green] (3.2, 10.2) circle (2pt); \shade[ball color=green] (3.75, 10.2) circle (2pt); \shade[ball color=green] (3.75, 10.8) circle (2pt); \draw (3,10.8) node {\small 1}; \draw (3,10.2) node {\small 4}; \draw (4,10.2) node {\small 2}; \draw (4,10.8) node {\small 3};
        \shade[ball color=green] (2.9, 9.7) circle (2pt); \shade[ball color=green] (2.9, 8.3) circle (2pt); \shade[ball color=green] (4.1, 8.3) circle (2pt); \shade[ball color=green] (4.4, 9) circle (2pt); \shade[ball color=green] (4.1, 9.7) circle (2pt); \draw (2.7,9.8) node {\small 1}; \draw (2.7,8.2) node {\small 4}; \draw (4.3,8.2) node {\small 5}; \draw (4.6,9) node {\small 2}; \draw (4.3,9.8) node {\small 3};
        \shade[ball color=green] (3, 7.75) circle (2pt); \shade[ball color=green] (2.4, 6.75) circle (2pt); \shade[ball color=green] (3, 5.75) circle (2pt); \shade[ball color=green] (4, 5.75) circle (2pt); \shade[ball color=green] (4.6, 6.75) circle (2pt); \shade[ball color=green] (4, 7.75) circle (2pt); \draw (2.8,7.85) node {\small 1}; \draw (2.2,6.75) node {\small 4}; \draw (2.8,5.65) node {\small 5}; \draw (4.2,5.65) node {\small 2}; \draw (4.8,6.75) node {\small 6}; \draw (4.2,7.85) node {\small 3};
        \shade[ball color=green] (2.95, 5.2) circle (2pt); \shade[ball color=green] (2.4, 4.25) circle (2pt); \shade[ball color=green] (2.95, 3.3) circle (2pt); \shade[ball color=green] (4.1, 3.3) circle (2pt); \shade[ball color=green] (4.57, 4.5) circle (2pt); \shade[ball color=green] (4.57, 4) circle (2pt); \shade[ball color=green] (4.1, 5.2) circle (2pt); \draw (2.75,5.3) node {\small 1}; \draw (2.2,4.25) node {\small 4}; \draw (2.75,3.2) node {\small 7}; \draw (4.35,3.2) node {\small 5}; \draw (4.8,4) node {\small 2}; \draw (4.8,4.5) node {\small 6}; \draw (4.35,5.3) node {\small 3};
	
        \shade[ball color=green] (2.95, 2.75) circle (2pt); \shade[ball color=green] (2.25, 2) circle (2pt); \shade[ball color=green] (2.25, 1) circle (2pt); \shade[ball color=green] (2.95, 0.25) circle (2pt); \shade[ball color=green] (4.05, 0.25) circle (2pt); \shade[ball color=green] (4.75, 2) circle (2pt); \shade[ball color=green] (4.75, 1) circle (2pt); \shade[ball color=green] (4.05, 2.75) circle (2pt); \draw (2.7,2.8) node {\small 1}; \draw (1.95,2) node {\small 4}; \draw (1.95,1) node {\small 7}; \draw (2.7,0.15) node {\small 5}; \draw (4.3,0.15) node {\small 2}; \draw (5,1) node {\small 6}; \draw (5,2) node {\small 8}; \draw (4.3,2.8) node {\small 3};
        \shade[ball color=green] (2.8, -0.3) circle (2pt); \shade[ball color=green] (2, -1.2) circle (2pt); \shade[ball color=green] (2,-2.3) circle (2pt); \shade[ball color=green] (2.8, -3.2) circle (2pt); \shade[ball color=green] (4.2, -3.2) circle (2pt); \shade[ball color=green] (5, -2.3) circle (2pt); \shade[ball color=green] (5.1, -1.75) circle (2pt); \shade[ball color=green] (5, -1.2) circle (2pt); \shade[ball color=green] (4.2, -0.3) circle (2pt); \draw (2.6,-0.2) node {\small 1}; \draw (1.8,-1.2) node {\small 4}; \draw (1.8,-2.3) node {\small 7}; \draw (2.6,-3.3) node {\small 5}; \draw (4.4,-3.3) node {\small 2}; \draw (5.2,-2.3) node {\small 6}; \draw (5.3,-1.75) node {\small 8}; \draw (5.2,-1.2) node {\small 3}; \draw (4.4,-0.2) node {\small 9};
        \shade[ball color=green] (2.9, -3.75) circle (2pt); \shade[ball color=green] (2.2, -4.3) circle (2pt); \shade[ball color=green] (1.9,-5.25) circle (2pt); \shade[ball color=green] (2.2,-6.15) circle (2pt); \shade[ball color=green] (2.9, -6.75) circle (2pt); \shade[ball color=green] (4.8, -6.15) circle (2pt); \shade[ball color=green] (4.1, -6.75) circle (2pt); \shade[ball color=green] (5.1, -5.25) circle (2pt); \shade[ball color=green] (4.8, -4.3) circle (2pt); \shade[ball color=green] (4.1, -3.75) circle (2pt); \draw (2.7,-3.65) node {\small 1}; \draw (2,-4.3) node {\small 4}; \draw (1.7,-5.25) node {\small 7}; \draw (2,-6.15) node {\small 5}; \draw (2.7,-6.85) node {\small 2}; \draw (4.3,-6.85) node {\small 6}; \draw (5,-6.15) node {\small 8}; \draw (5.3,-5.25) node {\small 3}; \draw (5.1,-4.3) node {\small 10}; \draw (4.3,-3.65) node {\small 9};
        \shade[ball color=green] (8.2, 11.8) circle (2pt); \shade[ball color=green] (8.2, 11.2) circle (2pt); \shade[ball color=green] (8.9, 11.5) circle (2pt); \draw (8,11.8) node {\small 1}; \draw (8,11.2) node {\small 2}; \draw (9.1,11.5) node {\small 3};
        \shade[ball color=green] (8.2, 10.8) circle (2pt); \shade[ball color=green] (8.2, 10.2) circle (2pt); \shade[ball color=green] (8.75, 10.2) circle (2pt); \shade[ball color=green] (8.75, 10.8) circle (2pt); \draw (8,10.8) node {\small 1}; \draw (8,10.2) node {\small 2}; \draw (9,10.2) node {\small 3}; \draw (9,10.8) node {\small 4};
        \shade[ball color=green] (7.9, 9.7) circle (2pt); \shade[ball color=green] (7.9, 8.3) circle (2pt); \shade[ball color=green] (9.1, 8.3) circle (2pt); \shade[ball color=green] (9.4, 9) circle (2pt); \shade[ball color=green] (9.1, 9.7) circle (2pt); \draw (7.7,9.8) node {\small 1}; \draw (7.7,8.2) node {\small 5}; \draw (9.3,8.2) node {\small 2}; \draw (9.6,9) node {\small 3}; \draw (9.3,9.8) node {\small 4};
        \shade[ball color=green] (8, 7.75) circle (2pt); \shade[ball color=green] (7.4, 6.75) circle (2pt); \shade[ball color=green] (8, 5.75) circle (2pt); \shade[ball color=green] (9, 5.75) circle (2pt); \shade[ball color=green] (9.6, 6.75) circle (2pt); \shade[ball color=green] (9, 7.75) circle (2pt); \draw (7.8,7.85) node {\small 1}; \draw (7.2,6.75) node {\small 5}; \draw (7.8,5.65) node {\small 2}; \draw (9.2,5.65) node {\small 6}; \draw (9.8,6.75) node {\small 3}; \draw (9.2,7.85) node {\small 4};
        \shade[ball color=green] (7.95, 5.2) circle (2pt); \shade[ball color=green] (7.4, 4.25) circle (2pt); \shade[ball color=green] (7.95, 3.3) circle (2pt); \shade[ball color=green] (9.1, 3.3) circle (2pt); \shade[ball color=green] (9.57, 4.5) circle (2pt); \shade[ball color=green] (9.57, 4) circle (2pt); \shade[ball color=green] (9.1, 5.2) circle (2pt); \draw (7.75,5.3) node {\small 1}; \draw (7.2,4.25) node {\small 5}; \draw (7.75,3.2) node {\small 2}; \draw (9.35,3.2) node {\small 6}; \draw (9.8,4) node {\small 3}; \draw (9.8,4.5) node {\small 7}; \draw (9.35,5.3) node {\small 4};
		
        \shade[ball color=green] (7.95, 2.75) circle (2pt); \shade[ball color=green] (7.25, 2) circle (2pt); \shade[ball color=green] (7.25, 1) circle (2pt); \shade[ball color=green] (7.95, 0.25) circle (2pt); \shade[ball color=green] (9.05, 0.25) circle (2pt); \shade[ball color=green] (9.75, 2) circle (2pt); \shade[ball color=green] (9.75, 1) circle (2pt); \shade[ball color=green] (9.05, 2.75) circle (2pt); \draw (7.7,2.8) node {\small 1}; \draw (6.95,2) node {\small 5}; \draw (6.95,1) node {\small 2}; \draw (7.7,0.15) node {\small 6}; \draw (9.3,0.15) node {\small 3}; \draw (10,1) node {\small 7}; \draw (10,2) node {\small 4}; \draw (9.3,2.8) node {\small 8};
        \shade[ball color=green] (7.8, -0.3) circle (2pt); \shade[ball color=green] (7, -1.2) circle (2pt); \shade[ball color=green] (7,-2.3) circle (2pt); \shade[ball color=green] (7.8, -3.2) circle (2pt); \shade[ball color=green] (9.2, -3.2) circle (2pt); \shade[ball color=green] (10, -2.3) circle (2pt); \shade[ball color=green] (10.1, -1.75) circle (2pt); \shade[ball color=green] (10, -1.2) circle (2pt); \shade[ball color=green] (9.2, -0.3) circle (2pt); \draw (7.6,-0.2) node {\small 1}; \draw (6.8,-1.2) node {\small 5}; \draw (6.8,-2.3) node {\small 2}; \draw (7.6,-3.3) node {\small 6}; \draw (9.4,-3.3) node {\small 3}; \draw (10.2,-2.3) node {\small 7}; \draw (10.3,-1.75) node {\small 4}; \draw (10.2,-1.2) node {\small 8}; \draw (9.4,-0.2) node {\small 9};
        \shade[ball color=green] (7.9, -3.75) circle (2pt); \shade[ball color=green] (7.2, -4.3) circle (2pt); \shade[ball color=green] (6.9,-5.25) circle (2pt); \shade[ball color=green] (7.2,-6.15) circle (2pt); \shade[ball color=green] (7.9, -6.75) circle (2pt); \shade[ball color=green] (9.8, -6.15) circle (2pt); \shade[ball color=green] (9.1, -6.75) circle (2pt); \shade[ball color=green] (10.1, -5.25) circle (2pt); \shade[ball color=green] (9.8, -4.3) circle (2pt); \shade[ball color=green] (9.1, -3.75) circle (2pt); \draw (7.7,-3.65) node {\small 1}; \draw (7,-4.3) node {\small 5}; \draw (6.7,-5.25) node {\small 2}; \draw (7,-6.15) node {\small 10}; \draw (7.7,-6.85) node {\small 6}; \draw (9.3,-6.85) node {\small 3}; \draw (10,-6.15) node {\small 7}; \draw (10.3,-5.25) node {\small 4}; \draw (10,-4.3) node {\small 8}; \draw (9.3,-3.65) node {\small 9};
        \draw (13, 11.5) node{$G_3^{AB}=\emptyset$};
        \shade[ball color=green] (13, 10.5) circle (2pt); \draw (13, 10.2) node{4};
        \draw (12.5, 9) -- node[above] {\tiny $ A\xleftarrow{1} B$} node[below] {\tiny $ A\xrightarrow{2} B$} (13.5, 9); \shade[ball color=green] (12.5, 9) circle (2pt); \shade[ball color=green] (13.5, 9) circle (2pt); \draw (12.3, 8.8) node{4}; \draw (13.7, 8.8) node{5};
        \draw (12.5, 6.75) -- (13.5, 6.75); \shade[ball color=green] (12.5, 6.75) circle (2pt); \shade[ball color=green] (13.5, 6.75) circle (2pt); \draw (12.3,6.55) node{4}; \draw (13.7,6.55) node{5};
        \draw (12.5, 4.75) -- (13.5, 4.75); \draw (12.5, 4.75) -- node[below, sloped] {\tiny $ A\xrightarrow{2} B$} (13, 4); \draw (13, 4) -- node[below, sloped] {\tiny $ A\xleftarrow{2} B$} (13.5,4.75); \shade[ball color=green] (12.5, 4.75) circle (2pt); \shade[ball color=green] (13.5, 4.75) circle (2pt); \shade[ball color=green] (13, 4) circle (2pt); \draw (12.3,4.95) node{4}; \draw (13.7,4.95) node{5}; \draw (13,3.65) node{7};
        \draw (12.5, 1.5) -- (13.5, 1.5); \draw (12.5, 1.5) -- (13, 0.75); \draw (13, 0.75) -- (13.5,1.5); \shade[ball color=green] (12.5, 1.5) circle (2pt); \shade[ball color=green] (13.5, 1.5) circle (2pt); \shade[ball color=green] (13, 0.75) circle (2pt); \shade[ball color=green] (13, 2.25) circle (2pt); \draw (12.3,1.7) node{4}; \draw (13.7,1.7) node{5}; \draw (13,0.45) node{7}; \draw (13,2.45) node{8};
        \draw (12.5, -1.75) -- (13.5, -1.75); \draw (12.5, -1.75) -- (13, -2.5); \draw (13, -2.5) -- (13.5,-1.75); \draw (12.5, -1.75) -- node[above, sloped] {\tiny $ A\xrightarrow{1} B$} (12.5,-1); \draw (12.5, -1) -- node[above, sloped] {\tiny $ A\xleftarrow{2} B$} (13.5,-1); \shade[ball color=green] (12.5, -1.75) circle (2pt); \shade[ball color=green] (13.5, -1.75) circle (2pt); \shade[ball color=green] (13, -2.5) circle (2pt); \shade[ball color=green] (12.5, -1) circle (2pt); \shade[ball color=green] (13.5, -1) circle (2pt); \draw (12.3,-2) node{4}; \draw (13.7,-2) node{5}; \draw (13,-2.8) node{7}; \draw (12.3,-0.7) node{9}; \draw (13.7,-0.7) node{8};
        \draw (12.5, -5.75) -- (13.5, -5.75); \draw (12.5, -5.75) -- (13, -6.5); \draw (13, -6.5) -- (13.5,-5.75); \draw (12.5, -5.75) -- (12.5,-5); \draw (12.5, -5) --(13.5,-5); \draw (12.5, -4.25) --node[above, sloped] {\tiny $ A\xrightarrow{2} B$} (13.5,-5); \shade[ball color=green] (12.5, -5.75) circle (2pt); \shade[ball color=green] (13.5, -5.75) circle (2pt); \shade[ball color=green] (13, -6.5) circle (2pt); \shade[ball color=green] (12.5, -5) circle (2pt); \shade[ball color=green] (13.5, -5) circle (2pt); \shade[ball color=green] (12.5, -4.25) circle (2pt); \draw (12.3,-6) node{4}; \draw (13.7,-6) node{5}; \draw (13,-6.8) node{7}; \draw (12.3,-4.75) node{9}; \draw (13.7,-4.75) node{8}; \draw (12.2,-4) node{10};
      \end{tikzpicture}
    \end{center}
    \caption{Cycles $A_\sq$ and $B_\sq$ and corresponding $G_\sq^{AB}$ }
    \label{fig.1}
  \end{figure}

  \begin{description}
  \item[$n=3$:] As mentioned above, $G^{AB}_{3}$ is a graph without vertices.

  \item[$n=4$:] Alice inserts her new node~$4$ between into the cycle-edge $\{1,2\}$ of her cycle~$A_3$; Bob inserts his new node~$4$ into the cycle-edge $\{1,3\}$ of his cycle~$B_3$.  Hence, $\{1,2\} = \nu_A(4) \ne \nu_B(4) = \{1,3\}$, so vertex~$4$ is added to $G_3^{AB}$.

  \item[$n=5$:] Alice inserts her new node~$5$ into the cycle-edge $\{2,4\}$ of her cycle~$A_4$; Bob inserts his new node~$5$ into the cycle-edge $\{1,2\}$ of his cycle~$B_4$.  Since $\{2,4\} = \nu_A(5) \ne \nu_B(5) = \{1,2\}$, vertex~$5$ is added to $G_4^{AB}$.  Let us check the edges:
    \begin{itemize}
    \item In $B_4$, the segment between $2$ and~$4$ contains the node~3 which is smaller than~$4$.  By Remark~\ref{rem:check-if-pair-is-a-nu}, there is no~$k$ with $\nu_A(5) = \nu_B(k)$, and hence no type-1 edge from $A$ to~$B$ at this time.
    \item As $\nu_B(5) = \nu_A(4)$, there is a type-1 edge between $4$ and~$5$ from $B$ to~$A$.
    \item Since $\max \nu_A(5)=4$ and $\nu_B(4) =\{1,3\} \not\ni 2$, there is also a type-2 edge between $5$ and~$4$ from $A$ to~$B$.
    \item Since $\max \nu_B(5)=2$ and $\nu_A(2) =\{1\} \ni 1$, there is no type-2 edge incident to $5$ from $B$ to $A$. 
    \end{itemize}

  \item[$n=6$:] Alices inserts her new node~$6$ into the cycle-edge $\{2,3\}$ of her cycle, Bob inserts his new node~6 into the cycle-edge $\{2,3\}$ of his cycle.  Since $\{2,3\} = \nu_A(6) = \nu_B(6) = \{2,3\}$, we don't have a vertex $6$ in $G_6^{AB}$.

  \item[$n=7$:] Alice throws into $\{4,5\}$, Bob throws into $\{3,4\}$.  Since $\{4,5\} = \nu_A(7) \ne \nu_B(7) = \{3,4\}$, the vertex~$7$ is added to $G_6^{A,B}$.
    \begin{itemize}
    \item In $B_6$, the segment between $4$ and~$5$ contain nodes~3 and~2 which are smaller than~$5$.  By Remark~\ref{rem:check-if-pair-is-a-nu}, there is no~$k$ with $\nu_A(7) = \nu_B(k)$, and hence no type-1 edge from $A$ to~$B$ at this time. 
    \item In $A_6$, the segment between $3$ and~$4$ contains the node~2 which is smaller than~$4$.  Again by Remark~\ref{rem:check-if-pair-is-a-nu}, there is no~$k$ with $\nu_B(7) = \nu_A(k)$, and thus no type-1 edge from $B$ to~$A$. 
    \item As $\max \nu_A(7)=5$ and $\nu_B(5)=\{1,2\} \not\ni 4$, we have a type-2 edge from $A$ to~$B$ between $7$ and~$5$.
    \item As $\max \nu_B(7)=4$ and $\nu_A(4)=\{1,2\} \not\ni 3$, there is also a type-2 edge from $B$ to~$A$ between $7$ and~$4$.
    \end{itemize}

  \item[$n=8$:] Alice plays $\{3,6\}$, Bob chooses $\{1,4\}$.  Since $\{3,6\} = \nu_A(8) \ne \nu_B(8) = \{1,4\}$, the vertex~$8$ is added to $G_7^{A,B}$.
    \begin{itemize}
    \item In~$B_7$, the segment between $3$ and~$6$ is empty (just the cycle-edge).  By Remark~\ref{rem:check-if-pair-is-a-nu}, there is no~$k$ with $\nu_A(8) = \nu_B(k)$, and hence no type-1 edge from $A$ to~$B$.
    \item For the same reason (segment between $1$ and~$4$ empty), there is no type-1 edge from $B$ to~$A$ incident to the vertex~$8$.
    \item $\max \nu_A(8)=6$ and $\nu_B(6)=\{2,3\} \ni 3$.  So there is no type-2 edge from $A$ to $B$ between $8$ and a smaller vertex.
    \item $\max \nu_B(8)=4$ and $\nu_A(4)=\{1,2\} \ni 1$.  So there is no type-2 edge from $B$ to $A$ between $8$ and a smaller vertex.
    \end{itemize}
    Hence, vertex $8$~is isolated in $G_8$.

  \item[$n=9$:] Alice chooses $\{1,3\}$, Bob chooses $\{1,8\}$.  As $\{1,3\} = \nu_A(9) \ne \nu_B(9) = \{1,8\}$, the vertex~$9$ is added to $G_8^{A,B}$.
    \begin{itemize}
    \item As $\{1,3\} = \nu_A(9) = \nu_B(4) = \{1,3\}$, there is a type-1 edge from $A$ to~$B$ between $9$ and~$4$.
    \item The segment between $1$ and $8$, contains the node $3$ which is smaller then $8$ so there is no type-1 edge from $B$ to~$A$ incident to the vertex~$9$. 
    \item As $\max \nu_A(9)=3$ and $\nu_B(3) = \{1,2\} \ni 1$, there is no type-2 edge from $A$ to~$B$ between $9$ and~$3$. 
    \item As $\max \nu_B(9)=8$ and $\nu_A(8) = \{3,6\} \not\ni 1$, there is a type-2 edge from $B$ to~$A$ between $9$ and~$8$.
    \end{itemize}

  \item[$n=10$:] Alice chooses $\{3,9\}$, Bob chooses $\{2,6\}$.  Since $\{3,9\} = \nu_A(10) \ne \nu_B(10) = \{2,6\}$, the vertex $10$ is added to $G_9^{A,B}$.
    \begin{itemize}
    \item The segment between $3$ and $9$ in $B$, contains the node $4$ which is smaller then $9$ so there is no type-1 edge from $A$ to~$B$ incident to the vertex~$10$. 
    \item Again, in~$B$, the segment between 3 and~9 has a vertex (4) smaller than~9.  Remark~\ref{rem:check-if-pair-is-a-nu} gives us that there is no~$k$ $\nu_B(k) = \nu_A(10)$, so no type-1 edge from $B$ to~$A$ is created.
    \item As $\max \nu_A(10)=9$ and $\nu_B(9) = \{1,8\}\not\ni 3$, there is a type-2 edge from $A$ to~$B$ between $10$ and~$9$.
    \item As $\max \nu_B(10)=6$ and $\nu_A(6) = \{2,3\} \ni 2$, no type-2 edge from $B$ to~$A$ is created.
    \end{itemize}
  \end{description}
\end{example}

\subsection{Rephrasing Theorem \ref{thm:main-polytope}}
We now rephrase Theorem~\ref{thm:main-polytope}, in terms of the pedigree graph.  We also unravel the little-o, and move to the ``Alice-and-Bob'' letters for the cycles.

\begin{theorem}[Theorem~\ref{thm:main-polytope}, rephrased]\label{thm:main:pedigreegraph}
  For every $\eps >0$ there is an integer~$N$ such that for all $n\ge N$ and all cycles $A_n$ with node set $[n]$, if $B_n$ is drawn uniformly at random from all cycles with node set~$[n]$, then
  \begin{equation*}
    \Prb\bigl( \text{ $G^{AB}_n$ is connected } \bigr) \ge 1-\eps.
  \end{equation*}
\end{theorem}

In symbols, and using infinite cycles, this reads:
\begin{equation*}
  \forall \eps>0 \;\; \exists N\colon \; \forall A \, \forall n\ge N\colon \Prb( \text{ $G^{AB}_n$ is connected } ) \ge 1-\eps,
\end{equation*}
where the probability is taken over all infinite cycles, see the next section.  A close look at our proof shows that we are actually proving the following stronger statement (we don't have any use for it, though):
\begin{equation*}
  \forall \eps>0 \;\; \exists N\colon \; \forall A\colon \; \Prb\bigl( \text{ $\forall n\ge N$: $G^{AB}_n$ is connected } \bigr) \ge 1-\eps.
\end{equation*}

\section{Pedigree Graphs of Random Cycles}\label{Sec:rnd-cyc-basics}
We have to reconcile uniformly random cycles with the ``evolution over time'' concept of pedigrees.  The definition of an infinite cycle makes that very convenient, just do the same as with infinite sequences of coin tosses: Take, as probability measure on the sample space $\prod_{n=1}^\infty [n]$ of all infinite cycles the product of the uniform probability measures on each of the sets $[n]$, $n\ge 3$.  We refer to the atoms in this probability space as \textit{random infinite cycles.}  The following is a basic property of product probability spaces.  We will use it mostly without mentioning it.

\begin{fact}\label{fact:Bn-uniform}
  If $B$ is a random infinite cycle, then, for each $n\ge 3$, the cycle $B_n$ is uniformly random in the set of all cycles with node set $[n]$.
\end{fact}

\paragraph{Creating isolated vertices.} %
The first substantial result about the connectedness of the pedigree graph, concerns the creation of isolated vertices.

As outlined in the introduction, we study the situation in which Alice chooses her cycle-edge of~$A_{n-1}$ according to a sophisticated strategy, whereas Bob always chooses a uniformly random cycle-edge of $B_{n-1}$ to insert his node~$n$ into (which amounts to his cycle~$B_n$ being uniformly random in the set of all cycles on $[n]$, by Fact~\ref{fact:Bn-uniform}).  In this section, we adopt a purely ``random graph'' perspective.  For fixed~$A$ and random~$B$, the pedigree graphs $G^{AB}_{\sq}$ are a sequence of random graphs, with some weirdo distribution: At time~$n$, whether the new vertex~$n$ is added or not, and if it is, how many incident edges it has, and what their end vertices are --- these are all random events/variables.

For deterministic~$A$ and random~$B$, let the random variable~$Y$ count the total number of times that an isolated vertex of the pedigree graph is created.\label{txt:def-of-Y} In other words, $Y = \sum_{n=4}^\infty \I{I_n}$, where $I_n$ denotes the event that, at time~$n$, $n$ is added as an isolated vertex to $G_n^{A,B}$ (and $\I{\sq}$ is the indicator random variable of the event).

\begin{lemma}\label{lem:isol-vert}
  Whatever Alice does, $\Exp Y = 2$.

  Moreover, for every $\eps>0$, if $n_0 \ge 4/\eps +2$, then, whatever Alice does
  \begin{equation*}
    \Prb\Bigl( \bigcup_{n \ge n_0} I_n \Bigr)  \quad \le \; \eps.
  \end{equation*}
\end{lemma}

(We have collected the proofs for this section in the following section, in order not to interrupt the motivating explanations.)

To understand why the lemma is important, consider a pedigree graph at time $n$, just before Alice and Bob make their choices of cycle-edges into which their respective new nodes~$n$ are inserted.  If~$n$ is not a vertex of the new pedigree graph~$G_n$, the number of connected components of~$G_{\sq}$ doesn't change.  If~$n$ is a vertex, and and it does have incident edges, then the number of connected components can only decrease.  The only way that the number of connected components of~$G_n$ can increase is if~$n$ is an isolated vertex in the new pedigree graph.  Hence, Lemma~\ref{lem:isol-vert} provides an upper bound on the expected number of connected components, uniform over~$n$.

\paragraph{The Intuition.}\label{paragraph:intuition} %
From Lemma~\ref{lem:isol-vert}, it is unlikely that the pedigree graph will have many components.  Indeed, intuitively, if only~2 isolated vertices are ever created, that means that most of the time either nothing happens (no new vertex) or edges are created, ultimately reducing the number of components, so the pedigree graph is connected.

While this basic intuition is essentially correct, a closer look reveals some subtleties.  First of all, Alice has a big sway in choosing the end vertices of new edges: she can pick the end vertices of type-2 edges from $A$ to~$B$; and she can influence the end vertices of type-1 edges (both directions).  

Secondly, Bob's choices are reduced by the low degrees of the vertices.  (A stronger version of~(a) is proved as Lemma~\ref{lem:segment-condition} in Section~\ref{pf-sec:basics}; we do not prove the rest of the lemma because we do not need it for our proof.)

\begin{lemma}
  The maximum degree of a vertex in a pedigree graph is at most~6:
  \begin{enumerate}[label=(\alph*),nosep]
  \item up to~2 to vertices created in the past; and
  \item up to~6 to future vertices.
  \end{enumerate}
\end{lemma}

Hence, if a vertex~$n_0$ was created as an isolated vertex or landed in a small connected component, Bob has only 4--6 shots at connecting it to another connected component.
The good news is that Alice can never ``shut down'' a connected component completely: Bob can always extend it by one more vertex.

\begin{lemma}\label{lem:connect-to-component}
  Let $C$ be a connected component of the pedigree graph $G^{AB}_{n-1}$.  There exists a $k\in C$ such that, no matter what Alice's move is at time~$n$, Bob has a move which creates the vertex~$n$ and makes it adjacent to~$k$.
\end{lemma}

However, for Bob to make a disconnected pedigree graph connected, at some time, he will have to manage to insert his new node in such a way that it has two incident edges, linking two connected components at the same time.

There is no difficulty in realizing that Alice wouldn't stand a chance against a strategically playing Bob.  But we claim that the game between a clever Alice and a blindfolded Bob will turn in Bob's favour almost all of the time.

\smallskip\noindent%
Computer simulations give another indication that some care has to be taken implementing the basic intuition: Even for~$n$ as large as~100, even if Alice's cycle is chosen uniformly \textsl{at random} instead of adversarial, the frequency (in 100000 samples) with which we saw a connected pedigree graph was only about 84\%.  In the remaining 16\% of cases, the typical situation is that of one giant connected component containing almost every vertex, and one tiny component growing only very slowly.
%
This indicates that even a \textsl{disinterested} Alice can do some damage.

\section{Basic Properties of the Pedigree Graph}\label{pf-sec:basics}
\begin{lemma}\label{lem:segment-condition}
  Let $A,B$ be two infinite cycles, and $n\ge 4$.
  \begin{enumerate}[label=\alph{*}.]
  \item\label{lem:segment-condition:cycle-edge---no-ped-edge} If $\nu_A(n)$ is a cycle-edge of $B_{n-1}$, then, if~$n$ is a vertex, there is no edge in $G^{AB}_n$ ``from $A$ to~$B$'' incident on~$n$.
  \item\label{lem:segment-condition:no-cycle-edge---ped-edge} If $\nu_A(n)$ is not a cycle-edge of $B_{n-1}$, then~$n$ is a vertex, and the pedigree graph $G^{AB}_n$ has an edge ``from $A$ to~$B$'' incident on~$n$:
    \begin{enumerate}[label*=\arabic{*}.]
    \item There is a type-1 edge, if and only if every node in the segment on~$B_{\sq}$ between $\nu^+_A(n)$ and $\nu^-_A(n)$ is larger than these two\footnote{Remember that segments are ``open'': they don't include the end nodes.}.  In this case, the other end of the type-1 edge is smallest node in the segment on~$B$ between $\nu^+_A(n)$ and $\nu^-_A(n)$.
    \item There is a type-2 edge, if and only if there is a node in the segment on~$B_{\sq}$ between $\nu^+_A(n)$ and $\nu^-_A(n)$ which is less than at least one of the two nodes.  In this case, the other end of the type-2 edge is $\max(\nu^+_A(n),\nu^-_A(n))$.
    \end{enumerate}
  \end{enumerate}
  In particular, there can be at most one edge ``from $A$ to~$B$'' between~$n$ and a vertex $k<n$.
\end{lemma}
\begin{proof}
  First of all, that~$n$ is a vertex follows immediately, as $\nu_A(n) = \nu_B(n)$ is not possible.

  The statements about the edges follow immediately form the definitions of the edges of the pedigree graph: For the first item, use Remark~\ref{rem:check-if-pair-is-a-nu}; for the second item, by the definition of ``segment between'', the two nodes $\nu^+_A(n)$ and $\nu^-_A(n)$ are separated on~$B$ by nodes smaller than themselves, so that none of them can be in the $\nu_B(\cdot)$-set of the other, by Remark~\ref{rem:past-nu}.
\end{proof}

\begin{lemma}\label{lem:cond-isol-vert}
  Let $A,B$ be two infinite cycles, and $n\ge 4$.
  Then~$n$ is an isolated vertex in $G_n$ if, and only if, both
  \begin{enumerate}[label=(\arabic*)]
  \item $\nu_A(n)$ is a cycle-edge of~$B_n$, and
  \item $\nu_B(n)$ is a cycle-edge of $A_n$.
  \end{enumerate}
\end{lemma}
\begin{proof}
  That the conditions (1) and~(2) are sufficient for~$n$ to be an isolated vertex is readily verified: Since $\nu_A(n)$ is a cycle-edge of $B_n$, it must still have been a cycle-edge of~$B_{n-1}$; similarly $\nu_B(n)$ was a cycle-edge of $A_{n-1}$.  Type-1 edges are immediately excluded.  As for type-2 edges, the condition $\nu_A( \max \nu_B(n) ) \cap \nu_B(n)$ is trivially satisfied.

  For the necessity, suppose (by symmetry) that $\nu_A(n)$ is not a cycle-edge of $B_n$.  Then either it was a cycle-edge of $A_{\sq}$ at time $n-1$ and got destroyed, or it wasn't a cycle-edge of $A_{\sq}$ at time $n-1$ in the first place.

  In the former case, the cycle-edge must have been destroyed by through $\nu_B(n) = \nu_A(n)$.  But this means that $n$ is not a vertex of the pedigree graph.

  In the latter case, Lemma~\ref{lem:segment-condition} applies, and~$n$ is a vertex, but it is not isolated.
\end{proof}
%
%
%

\subsection{Proof of Lemma~\ref{lem:isol-vert}}\label{pf-sec:rnd-cyc-basics:isol-vert}
Recall from page~\pageref{txt:def-of-Y} that $I_n$ denotes the event that, at time~$n$, the vertex~$n$ is added as an isolated vertex to the pedigree graph.

\begin{lemma}\label{lem:prob-isol-vert}
  For~$n\ge4$, $\displaystyle \Prb( I_n ) = \frac{4}{(n-1)(n-2)}$.
\end{lemma}
\begin{proof}
  For $x=1,2$, denote by~$E_x$ the event of condition~($x$) of Lemma~\ref{lem:cond-isol-vert}.
  We need to compute $\displaystyle \Prb(E_1\cap E_2) = \Prb(E_1) \Prb(E_2 \mid E_1)$.
  As for $E_1$, we have computed this probability in the proof of Lemma~\ref{lem:common-edges}: it is $2/(n-1)$.

  As for $\Prb(E_2 \mid E_1)$, conditioning on $\nu_A(n)$ being a cycle-edge amounts to
  \begin{enumerate}[label=\arabic*.]
  \item identifying these two nodes to a (super-)node, leaving $n-1$ nodes to consider, and taking a uniformly random cycle on these $n-1$ nodes; and then
  \item replacing the super-node by the cycle-edge $\nu_A(n)$, which requires deciding which of the two nodes comes first (in positive orientation).
  \end{enumerate}

  To calculate the conditional probability $\Prb(E_2 \mid E_1)$, we go over all cycle-edges~$e$ of $A_n$, and find the conditional probability of the event that $\nu_B(n) = e$.  Since these events are mutually exclusive, we then just add up the probabilities.

  Of the $n-1$ nodes after forming the super-node, $n-2$ involve nodes different from~$n$.  Each one of the $\binom{n-2}{2}$ possible cycle-edges between these nodes has equal chance of ending up being the one chosen by Bob at time~$n-1$.  Let us start with the cycle-edges~$e$ of $A_n$ with $\nu_A(n)\cap e = \emptyset$, i.e., they are not incident to the super-node.  There are $n-4$ of them.  For each one of them, we have
  \begin{equation*}
    \Prb( e = \nu_B(n) \mid E_1 ) = 1/\binom{n-2}{2}.
  \end{equation*}

  The probability that Bob's chosen cycle-edge is one of the two cycle-edges incident on the super-node is $2/\binom{n-2}{2}$.  Let~$e$ be one of the two cycle-edges with $\abs{ \nu_A(n)\cap e } =1$.  The probability that the orientation of the cycle-edge $\nu_A(n)$ of $B_n$ chosen in step~(2) above is so that Bob's chosen cycle-edge is equal to~$e$ is $\nfrac12$.  In total, we have, for each of these~2 edges~$e$,
  \begin{equation*}
    \Prb( e = \nu_B(n) \mid E_1 )
    = \frac{\nfrac 12}{\binom{n-2}{2}}.
  \end{equation*}

  We add up:
  \begin{equation*}
    \Prb(E_2 \mid E_1)
    = \frac{n-4}{\binom{n-2}{2}} + \frac{\nfrac12}{\binom{n-2}{2}} + \frac{\nfrac12}{\binom{n-2}{2}}
    = \frac{n-3}{\binom{n-2}{2}}
    = \frac{2}{n-2}.
  \end{equation*}
  In total, we have
  \begin{equation*}
    \Prb( I_n )
    = \Prb(E_1) \Prb(E_2 \mid E_1)
    = \frac{2}{n-1} \cdot \frac{2}{n-2}
    = \frac{4}{(n-1)(n-2)}.
  \end{equation*}
\end{proof}

We can now complete the proof of Lemma~\ref{lem:isol-vert}.
\begin{proof}[of Lemma~\ref{lem:isol-vert}]
  For the expectation, we calculate
  \begin{align*}
    \Exp Y
    &= \sum_{n=4}^\infty \Prb( I_n ) \\
    &= \sum_{n=4}^\infty \frac{4}{(n-1)(n-2)} &&\comment{by Lemma~\ref{lem:prob-isol-vert}} \\
    &= 4 \sum_{n=4}^\infty \lt( \frac{1}{n-2} - \frac{1}{n-1} \rt) \\
    &= 4 \lim_{n\to\infty}\bigl( \nfrac{1}{2} - \nfrac{1}{(n-1)} \bigr)\\
    &= 2.
  \end{align*}

  Similarly, for the statement about $I_n$, we find that
  \begin{equation*}
    \Prb\Bigl( \bigcup_{n \ge n_0} I_n \Bigr)
    \le
    \sum_{n=n_0}^\infty \Prb(I_n)
    =
    4 \sum_{n=n_0}^\infty \lt( \frac{1}{n-2} - \frac{1}{n-1} \rt)
    =
    4/(n_0-2)
    \le
    \eps.
  \end{equation*}
  This completes the proof of Lemma~\ref{lem:isol-vert}.
\end{proof}

\subsection{Proof of Lemma~\ref{lem:connect-to-component}}
\begin{proof}[Proof of Lemma~\ref{lem:connect-to-component}]
  Take $k := \max V(C)$.
  Since $k$ is a vertex, we have $\abs{ \nu_B(k) \cap \nu_A(k) } \le 1$.  Suppose that $\nu^+_B(k) \notin \nu_A(k)$ (the other case is symmetric).  Then, the first time Bob inserts a node, say~$n'$, into the cycle-edge on the positive side of~$k$, this will create a type-2 edge ``from $B$ to~$A$'' between $n'$ and~$k$.  Since~$k$ is the newest vertex in its component, Bob has not yet inserted a node there, so he can insert~$n$ there, now.
\end{proof}

\section{The Adjacency Game}\label{Sec:cnct-game}
At each time, Alice moves first.  As already explained, she determines the cycle~$A$ by choosing, at each time~$n$, the cycle-edge of~$A_{n-1}$ into which her new node~$n$ will be inserted.  Then Bob moves.  He determines~$B$ in the same way, but (using Fact~\ref{fact:Bn-uniform}), he will draw the cycle-edge of $B_{n-1}$ into which his new node~$n$ is inserted uniformly at random from all cycle-edges of $B_{n-1}$, and his choice is independent of his earlier choices.

We say that Bob wins, if there exists an $n_0$ such that for all $n\ge n_0$, the pedigree graph $G^{AB}_n$ is connected.  We need Bob to win ``uniformly'', i.e., $n_0$ must not depend on Alice's moves.

Let the random variable $T_n$ denote the number of connected components in the pedigree graph $G^{AB}_n$.  To analyze the development of the random process $T_\sq$, it turns out to be useful to consider a second random process, $S_\sq$.  Denote by $\CmnE_n$ the set of cycle-edges that Alice's cycle and Bob's cycles have in common,
\begin{equation*}
  \CmnE_n := E(A_n) \cap E(B_n),
\end{equation*}
and let
\begin{equation*}
  S_n := \abs{\CmnE_n}
\end{equation*}
count the number of cycle-edges that Alice and Bob have in common.  We will distinguish Alice's moves by whether or not she chooses a common cycle-edge to place her new cycle node.  
The set $\CmnEa_n$ holds those common cycle-edges which are not incident on the cycle-edge which Alice chooses for her new node:
\begin{equation*}
  \CmnEa_n := \bigl\{ e \in \CmnE_n \mid e \cap \nu_A(n+1) = \emptyset \bigr\};
\end{equation*}
we let $\Sa$ count the cycle-edges in $\CmnEa$:
\begin{equation*}
  \Sa_n := \abs{\CmnEa_n}.
\end{equation*}
Finally, denote by $\RstE_n$ the set of cycle-edges in Bob's cycle which are neither common nor incident on Alice's chosen cycle-edge:
\begin{align*}
  \RstE_n   &:= \bigl\{ e \in E(B_n) \setminus \CmnE_n \mid e\cap \nu_A(n+1) = \emptyset \bigr\} \text{; and}\\
  \Rst_n &:= \abs{ \RstE_n }.
\end{align*}

We are now ready to state and prove the transition probabilities.  They depend on whether Alice chooses, for her new node, a common cycle-edge --- we refer to that as a c-move by Alice --- or a cycle-edge which is in the difference $E(A_n)\setminus \CmnE_n$ --- we call that a d-move.

\begin{lemma}\label{lem:transition-prbs}
  The conditional probabilities
  \begin{equation*}
    \Prb\Bigl(  S_{n+1}=S_n + \Delta_S \;\; \land \;\; T_{n+1}=T_n + \Delta_T \;\; \mid \; B_n \Bigr)
  \end{equation*}
  satisfy these bounds (entries not shown are ``$=0$''):
  \begin{center}
    \renewcommand{\arraystretch}{1.333333}
    \begin{tabular}{c|c|c|c|c|l c c|c|c|c|c|l}
      \multicolumn{1}{c|}{$\Delta_T$}&\multicolumn{5}{c}{}                                            &             &\multicolumn{1}{c|}{$\Delta_T$}&\multicolumn{5}{c}{} \\[.5ex]
      \cline{2-5}\cline{9-12}
      $+1$       &$=\frac{\Sa_n}{n}$&$\le\frac{2}{n}$   &                &                &          &             &$+1$       &           &                  &                           &                &                   \\[.25ex]
      \cline{2-5}\cline{9-12}
      $0$        &                  &$=\frac{\Rst_n}{n}$&$\le\frac{2}{n}$&$=\frac{1}{n}$  &          &             &$0$        &           &$=\frac{\Sa_n}{n}$&$\le\frac{\Rst_n-T_n+1}{n}$&$\le\frac{4}{n}$&                   \\[.25ex]
      \cline{2-5}\cline{9-12}
      $-1$       &                  &                   &                &                &          &             &$-1$       &\hspace*{2em}&                 &\multicolumn{2}{|c|}{ $\ge\frac{T_n-1}{n}$ }    &                   \\[.25ex]
      \cline{1-6}\cline{8-13}
      {}         &\multicolumn{1}{c}{$-2$}&\multicolumn{1}{c}{$-1$}&\multicolumn{1}{c}{$0$}&\multicolumn{1}{c}{$+1$}&$\Delta_S$&&{}&\multicolumn{1}{c}{$-2$}&\multicolumn{1}{c}{$-1$}&\multicolumn{1}{c}{$0$}&\multicolumn{1}{c}{$+1$}& $\Delta_S$        \\
      \multicolumn{6}{c}{c-move}                                                                     &\hspace{2em} &\multicolumn{6}{c}{d-move} \\[-1ex]
      \multicolumn{6}{c}{\footnotesize(Alice chooses common cycle-edge)}                                   &             &\multicolumn{6}{c}{\footnotesize(Alice chooses cycle-edge in $E(A_n)\setminus \CmnE_n$)}\\
    \end{tabular}
  \end{center}
\end{lemma}
\begin{proof}
  Let us start with the case that Alice makes a \textbf{c-move,} i.e., she chooses a common cycle-edge to insert her new node $n+1$ into.  In symbols: $\nu_A(n+1)\in \CmnE_n$.
  In this case, there can be no edges ``from $A$ to~$B$'' incident to the vertex $n+1$ of $G^{AB}_{n+1}$.  We split into disjoint events based on
  \begin{equation*}
    H := \abs{  \nu_A(n+1) \cap \nu_B(n+1)  }
  \end{equation*}

  \textit{Event $H=2$.} %
  This happens if Bob's random choice of a cycle-edge for his new node $n+1$ is the same as Alices, i.e., $\nu_A(n+1)$.  The probability of his happening is $\nfrac1n$.  In that case, both $S$ and~$T$ are unchanged.

  \textit{Event $H=1$.} We split into two sub-events, ``$H=1$ and $\nu_B(n+1)\in \CmnE_n$'' and ``$H=1$ and $\nu_B(n+1)\notin \CmnE_n$''.  Each has probability at most $2/n$.
  \begin{itemize}
  \item $H=1$ and $\nu_B(n+1)\in \CmnE_n$ implies that both $\nu_A(n+1)$ and $\nu_B(n+1)$ are still common cycle-edges after Alice and Bob have added their new nodes $n+1$.  By Lemma~\ref{lem:cond-isol-vert}, this implies that $n+1$ is an isolated vertex of~$G_{n+1}$.  We have $T_{n+1} = T_n +1$, and $S_{n+1} = S_n -1$ (two common cycle-edges destroyed, one new created).
  \item Suppose $H=1$ and $\nu_B(n+1)\notin \CmnE_n$.  Simply from $\nu_B(n+1)\notin \CmnE_n$, applying Lemma~\ref{lem:segment-condition}\ref{lem:segment-condition:no-cycle-edge---ped-edge} (with $A$/$B$ exchanged), the existence of an edge incident on $n+1$ ``from $B$ to~$A$'' follows.  Moreover, that lemma also states that there is only one such edge.  Hence, we have $T_{n+1} = T_n$, and $S_{n+1} = S_n$ (one common cycle-edge destroyed, one new created).
  \end{itemize}

  \textit{Event $H=0$.}  We distinguish the same two sub-events as for $H=0$: ``$H=0$ and $\nu_B(n+1)\in \CmnE_n$'' and ``$H=0$ and $\nu_B(n+1)\notin \CmnE_n$''.
  \begin{itemize}
  \item The sub-event $H=0$ and $\nu_B(n+1)\in \CmnE_n$ occurs if Bob hits a cycle-edge in $\CmnEa_n$.  The probability of that happening is $\Sa_n/n$.  By Lemma~\ref{lem:segment-condition}\ref{lem:segment-condition:cycle-edge---no-ped-edge}, we do not have an edge in the pedigree graph incident on $n+1$, so an isolated vertex is created.  Hence $T_{n+1} = T_n +1$, and $S_{n+1} = S_n -2$ (two common cycle-edges destroyed, none new created).
  \item That $H=0$ and $\nu_B(n+1)\notin \CmnE_n$ means that that Bob hits a cycle-edge in $\RstE_n$; the probability of this happening is $\Rst_n/n$.  As in the sub-event above, applying Lemma~\ref{lem:segment-condition}\ref{lem:segment-condition:no-cycle-edge---ped-edge} gives the existence of exactly one edge ``from $B$ to~$A$'' incident on $n+1$.  Hence, we have $T_{n+1} = T_n$, and $S_{n+1} = S_n -1$ (one common cycle cycle-edge destroyed, none new created).
  \end{itemize}

  This completely partitions the probability space, with probabilities corresponding to the bounds in the table.

  \mypar%
  Now let us consider the case that Alice makes a \textbf{d-move,} i.e., she chooses a cycle-edge which is not common to both cycles, to insert her new node $n+1$ into.  In symbols: $\nu_A(n+1)\notin \CmnE_n$.
  By Lemma~\ref{lem:segment-condition}, this condition implies that $n+1$ is a vertex of the pedigree graph $G^{AB}_{n+1}$ and there is an edge ``from $A$ to~$B$'' incident on $n+1$.  This implies that $T_{n+1} \le T_n$.

  We distinguish cases based on the event $\nu_B(n+1) \in \CmnE_n$.

  \textit{Event $\nu_B(n+1) \in \CmnE_n$.}  We split into three subevents, depending on $H$ (as defined above):
  \begin{itemize}
  \item $\nu_B(n+1) \in \CmnE_n$ and $H=0$: This means that Bob hits a cycle-edge in $\CmnEa_n$.  The probability of that happening is $\Sa_n/n$.  In that case, by Lemma~\ref{lem:segment-condition}\ref{lem:segment-condition:cycle-edge---no-ped-edge}, there is no edge ``from $B$ to~$A$'', and we have $T_{n+1}=T_n$.  One common cycle-edge was destroyed, but none was created, so $S_{n+1} = S_n-1$.
  \item $\nu_B(n+1) \in \CmnE_n$ and $H=1$: There are at most~2 cycle-edges in $B_n$ which are both common with $A_n$ and have a node in common with $\nu_A(n+1)$.  Hence, the probability that Bob hits one of them is $2/n$.  Again, $T_{n+1}=T_n$.  One common cycle-edge was destroyed, but another one is created, so $S_{n+1}=S_n$.
  \end{itemize}
  Note that $H=2$ is impossible in a d-move.

  \textit{Event $\nu_B(n+1) \notin \CmnE_n$.}  Before we can proceed, we need to extract the following fact from the proof of Lemma~\ref{lem:connect-to-component}:

 \begin{fact}\label{fact:reduce-components}
   If $G^{AB}_n$ is not connected and~$C$ is a connected component, then, no matter what Alice's move at time $n+1$, there is a cycle-edge $\{k,k'\}$ of $B_n$ with the properties
   \begin{enumerate}[nosep,label=(\alph*)]
   \item $k' < k$ and $k \in C$;
   \item if Bob inserts his new node $n+1$ into that cycle-edge, there will be a type-2 edge ``from $B$ to~$A$'' betwen $n+1$ and~$k$.
   \end{enumerate}
 \end{fact}
 We already know that the cycle-edge $\{k,k'\}$ is not in $\CmnE$ (Lemma~\ref{lem:segment-condition}\ref{lem:segment-condition:cycle-edge---no-ped-edge}).

 Of the $T_n$ connected components at time~$n$, one contains the end vertex of the edge ``from $A$ to~$B$'' incident on $n+1$.  For each of the remaining $T_n-1$ components, take one cycle-edge as described in Fact~\ref{fact:reduce-components}.  By property~(a), all these cycle-edges are distinct.  Hence, the chance that Bob hits one of them is exactly $(T_n-1)/n$.  There may be more possibilities for Bob to reduce the number of connected components, but, by the last sentence in Lemma~\ref{lem:segment-condition}, the number of components which can be joined at time $n+1$ is at most~2.

 Now we can proceed with the event $\nu_B(n+1) \notin \CmnE_n$.  We split into sub-events:
 \begin{itemize}
 \item $\nu_B(n+1) \notin \CmnE_n$ and $T_{n+1}=T_n-1$: By what we have discussed, the probability of this happening is at least $(T_n-1)/n$.  As for the change in $S_{\sq}$, no common cycle-edge is destroyed (Lemma~\ref{lem:segment-condition}\ref{lem:segment-condition:cycle-edge---no-ped-edge}), but it is possible that one is created, so $S_{n+1} \in \{S_n,S_n+1\}$.
 \item $\nu_B(n+1) \notin \CmnE_n$ and $T_{n+1}=T_n$: We split into sub-sub-events:
   \begin{itemize}
   \item The above conditions and $H=1$: The probability of that happening is at most $4/n$.  Since no common cycle-edge is destroyed, but a new one is created, we have $S_{n+1}=S_n +1$.
   \item The above conditions and $H=0$: For that to happen, Bob has to hit a cycle-edge in $\RstE_n$.  At least $T_n-1$ of these cycle-edges lead to the situation $T_{n+1}=T_n-1$, so the probability for this sub-sub-event is at least $\frac{\Rst_n-T_n+1}{n}$.  No common cycle-edge is destroyed, and none is created, so $S_{n+1} = S_n$.
   \end{itemize}
 \end{itemize}

 This completely partitions the probability space, with probabilities corresponding to the bounds in the table.
\end{proof}

The proof of the main theorem now follows the following idea.  From the tables in Lemma~\ref{lem:transition-prbs}, you see that d-moves have chance of reducing the number of connected components --- albeit a small one.  Moreover, Alice cannot take a c-move only when $S_n > 0$, but c-moves have a strong tendency to reduce $S_\sq$.  We prove that the number of d-moves that Alice has to take are frequent enough to lead to a decrease in the number of connected components.  This suffices to prove Theorem~\ref{Sec:mainproof}, along the lines sketched on page~\pageref{paragraph:intuition}.

The next section gives more details of the proof.

\section{Proof of Theorem~\ref{thm:main:pedigreegraph}}\label{Sec:mainproof}
Using the Azuma-Hoeffding super-martingale tail bound, we prove that, for large enough $n_0$, Alice has to take many d-moves between times $n_0$ and $2n_0$.

\begin{lemma}\label{lem:Alice-will-play--d--often}
  For every $\eps\in\lt]0,1\rt[$, if $n_0 \ge \max(900, 8\ln(1/\eps))$, and $n_1 := 2 n_0$ then, whatever Alice does, the probability, conditioned on $B_{n_0}$ and $S_{n_0}\le \ln^2 n_0$, that among her moves at times $n = n_0+1,\dots,n_1$, there are fewer than $n_0/3$ d-moves, is at most~$\eps$.
\end{lemma}
\begin{proof}
  Denote by $C \subseteq \{n_0+1,\dots,n_1\}$ the set of times in which Alice takes c-move; so Alice takes d-moves at every time $D := \{n_0+1,\dots,n_1\}\setminus C$.

  Denote by $C_n$ the event that Alice chooses a c-move at time $n+1$.  The reason for the index~$n$ instead of $n+1$ is that $C_n$ and $D_n$ are $B_n$ measureable\footnote{%
    This is the case because Alice's decision is based on $A_n$ and $B_n$ in a deterministic way, and $A_n$ is deterministically determined from $B_{n-1}$ which is determined by $B_n$, and so on (induction).%
  };%
  we need that for the Martingale argument below.

  Now, for every $n \in C$, let
  \begin{equation*}
    L_n := \I{\{S_n < S_{n-1}\}} - \I{\{S_n > S_{n-1}\}}.
  \end{equation*}
  Since, by Lemma~\ref{lem:transition-prbs}, when Alice decides for a c-move, $S_n > S_{n-1}$ is equivalent to $S_n = S_{n-1}+1$, we have
  \begin{equation}\label{eq:Alice-will-play-d-often:c-mv--bd-S}
    S_n \le S_{n-1} - L_n,
  \end{equation}
  and
  \begin{equation*}
    \Prb(S_{n+1} < S_n \mid B_n, C_n )
    \ge
    \frac{\Sa_n}{n} + \frac{\Rst_n}{n}
    \ge
    \frac{S_n-4}{n} + \frac{n-S_n-4}{n}
    =
    1-8/n;
  \end{equation*}
  note that the inequalities remain valid if $S_n < 4$.  As for increasing $S_{\sq}$, we have
  \begin{equation*}
    \Prb( S_{n+1} > S_n \mid B_n, C_n )  =  1/n.
  \end{equation*}
  Hence
  \begin{equation}\label{eq:Alice-will-play-d-often:Exp-of-L}
    \Exp\bigl(L_{n+1} \mid B_n, C_n \bigr) \ge 1-9/n \ge .99,
  \end{equation}
  where the last inequality follows since $n_0 \ge 900$.

  For the $n\in D$, i.e., when Alice decides for a d-move, we upper-bound by~1 the probability that Alices increases $S_{\sq}$, so,
  \begin{equation}\label{eq:Alice-will-play-d-often:d-mv--bd-S}
    S_{n+1} \le S_n + 1-\I{C_n},
  \end{equation}
  since, by Lemma~\ref{lem:transition-prbs}, $S_{n+1}>S_n$ is equivalent to $S_{n+1}=S_n+1$ for d-moves, too.

  Combining inequalities \eqref{eq:Alice-will-play-d-often:c-mv--bd-S} and~\eqref{eq:Alice-will-play-d-often:d-mv--bd-S}, we have
  \begin{equation}\label{eq:Alice-will-play-d-often:total--bd-S}
    S_{n+1} - S_n \le 1-\I{C_n} + \I{C_n}\cdot L_{n+1}.
  \end{equation}

  We define a super-martingale as follows.  Let, $X_{n_0} := S_{n_0}$ and for each $n=n_0,\dots,n_1-1$,
  \begin{equation*}
    X_{n+1} = X_n + \I{C_n}\cdot\bigl( .99 - L_{n+1} \bigr).
  \end{equation*}
  (We define $L_n := 0$ for $n\not\in C$.)

  By induction, $X_{n+1}$ is determined by $B_{n+1}$, so the measurability property of a super-martingale is given.
  Moreover, by~\eqref{eq:Alice-will-play-d-often:Exp-of-L}, we have
  \begin{equation*}
    \Exp\bigl( X_{n+1} \mid B_n \bigr)
    =
    X_n + \I{C_n}\cdot \bigl( .99 - \Exp(L_{n+1} \mid B_n) \bigr)
    \le
    X_n;
  \end{equation*}
  hence $X_\sq$ is in fact a super-martingale.

  \begin{claim}
    For $n \ge n_0$, we have $\displaystyle X_n \ge S_n + 1.99\sum_{j=n_0}^{n-1}C_j -(n-n_0)$.
  \end{claim}
  \textit{Proof of the claim, by induction.}  The claim holds for $n=n_0$.  Assume that it holds for~$n\ge n_0$.  We have
  \begin{align*}
    X_{n+1} &= X_n + \I{C_n}\cdot\bigl( .99 - L_{n+1} \bigr)
    \\
    &\ge S_n + 1.99\sum_{j=n_0}^{n-1}C_j -(n-n_0) \quad+ \I{C_n}\cdot\bigl( .99 - L_{n+1} \bigr) &&\comment{I.H.}
    \\
    &= S_n + 1.99\sum_{j=n_0}^{n-1}C_j -(n-n_0) \quad+1.99\cdot\I{C_n} -1 \quad+ (1-\I{C_n}) -\I{C_n} L_{n+1}
    \\
    &\ge S_n + 1.99\sum_{j=n_0}^{n-1}C_j -(n-n_0) \quad+1.99\cdot\I{C_n} -1 \quad+ S_{n+1}-S_n &&\comment{\eqref{eq:Alice-will-play-d-often:total--bd-S}}
    \\
    &= S_{n+1} + 1.99\sum_{j=n_0}^{n}C_j -(n+1-n_0),
  \end{align*}
  which completes the proof of the claim. \quad$\blackdiamond$

  \mypar%
  From the claim, since $S_n \ge 0$ always, we have
  \begin{equation*}
    X_{n_1} \ge 1.99\sum_{j=n_0}^{n_1-1}C_j -(n_1-n_0) = 1.99 \abs{C} - (n_1-n_0) = 1.99\abs{C} - n_0.
  \end{equation*}
  The event $\abs{D} \le n_0/3$ implies the event $\abs{C} \ge 2 n_0 /3$, and hence the event $X_{n_1} \ge 1.2 n_0$.

  To apply the Azuma-Hoeffding inequality for super-martingales (e.g., Lemma 4.2 in \cite{Wormald-LecNot-99}), we first note that $\abs{ X_{n+1} - X_n } \le 1.99$ deterministically.  We conclude that the probability of the event $\abs{D} \le n_0/3$ is at most
  \begin{equation*}
    \exp\lt(-\frac{ \bigl( 1.2 n_0 - S_{n_0} \bigr)^2 }{2\cdot 1.99^2\cdot n_0}  \rt)
    \le
    \exp\lt(-\frac{ \bigl( 1.2 n_0 - S_{n_0} \bigr)^2 }{8n_0}  \rt)
    \le
    \exp\lt(-\frac{ n_0^2 }{8 n_0}  \rt),
  \end{equation*}
  where the last inequality follows since $S_{n_0} \le \ln^2 n_0 \le .2 n_0$, whenever $n_0 \ge 300$.
  We continue, using $n_0 \ge 8 \ln(1/\eps)$,
  \begin{equation*}
    \Prb( \abs{D} \le n_0/3 )
    \le
    e^{-n_0/8}
    \le
    \eps,
  \end{equation*}
  which concludes the proof of the lemma.
\end{proof}

From this, we deduce that must $T_\sq$ decrease, but some sophistication is needed, because of the slow divergence of $\sum 1/n$: Indeed, between $n_0$ and $2n_0$, $T_\sq$ decreases only with a constant probability:

\begin{lemma}\label{lem:T-will-decrease}
  Fix $\delta := \nfrac{1}{42}$.  If $n_0 \ge \max(900, 8\ln(1/\delta))$, and $n_1 := 2 n_0$ then, whatever Alice does,
  \begin{equation*}
    \Prb\Bigl( \exists n \in \{n_0+1,\dots,n_1\}\colon \;\; T_{n+1} < T_n \;\; \Bigm| B_{n_0}, T_{n_0} \ge 2, S_{n_0} \le \ln^2 n_0 \Bigr)  \quad \ge \quad 1/7.
  \end{equation*}
\end{lemma}
\begin{proof}
  From Lemma~\ref{lem:Alice-will-play--d--often}, with probability at least $1-\delta$, there are at least $n_0/3$ d-moves among Alices moves at times $n=n_0+1,\dots,n_1:=2 n_0$.

  Denote by~$D$ the set of d-moves by Alice in time $\{n_0+1,\dots,n_1\}$.  By Lemma~\ref{lem:transition-prbs}, the probability~$p$ that none of these d-moves decreases $T_{\sq}$ is at most
  \begin{equation*}
    p
    \le
    \prod_{n \in D} \Prb\bigl( T_n < T_{n-1} \mid T_{n-1} \ge 2 \bigr)
    \le
    \prod_{n \in D} \lt( 1- \frac{1}{n} \rt).
  \end{equation*}
  Hence
  \begin{equation*}
    \ln p
    \le
    \sum_{n \in D} \ln(1-1/n)
    \le
    -\sum_{n \in D} 1/n.
  \end{equation*}
  The subset~$D$ of $\{n_0+1,\dots,n_1\}$ of size at least $n_0/3$ which maximizes the last expression is $D := \{5 n_0 /3,\dots,n_1\}$, so we get
  \begin{equation*}
    \ln p
    \le
    -\sum_{n=5 n_0 /3}^{n_1} 1/n
    \le
    -\ln\Bigl( \frac{n_1}{5 n_0 /3} \Bigr)
    =
    -\ln(6/5).
  \end{equation*}
  We conclude that $p \le 5/6$.  In total, the probability that $T_{\sq}$ never decreases is at most $5/6 + \delta = 5/6+1/42 = 6/7$.
\end{proof}

We can boost the probability to $1-\eps$, for arbitrary $\eps >0$, by iterating the argument $\Omega(\ln(1/\eps))$ times.

\begin{lemma}\label{lem:T-will-decrease:boost}
  Fix $\delta := \nfrac{1}{42}$.  For every $\eps\in\lt]0,\nfrac{1}{56}\rt[$, with $a := 10\ln(2/\eps)$, if
  \begin{equation*}
    n_0 \ge \max(900, 8\ln(1/\delta), (2a)^{4/\eps}, e^{6/\eps}),
  \end{equation*}
  and $n_1 := 2 a n_0$ then, whatever Alice does,
  \begin{equation*}
    \Prb\Bigl( \exists n \in \{n_0,\dots,n_1\}\colon \;\; T_{n+1} < T_n \;\; \Bigm| B_{n_0}, T_{n_0} \ge 2 \Bigr)  \quad \ge \quad 1-\eps
  \end{equation*}
\end{lemma}
Before we can prove Lemma~\ref{lem:T-will-decrease:boost}, we need to control the size of $S_\sq$ through the following two lemmas.

\begin{lemma}\label{lem:common-edges}
  If $n_0 \ge 4$, then, whatever Alice does,
  \begin{equation*}
    \Prb( S_{n_0} > \ln n_0 ) \le 3 / \ln n_0
  \end{equation*}
\end{lemma}
\begin{proof}
  For $e=\{k,\ell\} \in \binom{[{n_0}]}{2}$, denote by $E_{n_0}(e)$ the event that~$e$ is a cycle-edge of~$B_{n_0}$.  For ${n_0}\ge 4$, by Fact~\ref{fact:Bn-uniform},
  \begin{equation*}
    \Prb( E_{n_0}(e) ) = 2/({n_0}-1).
  \end{equation*}
  Since
  \begin{equation*}
    S_{n_0} = \sum_{e \in A_{n_0}} \Ind( E_{n_0}(e) ),
  \end{equation*}
  we find that $\Exp S_{n_0} = \frac{2n}{{n_0}-1} \le 3$ (the last inequality follows from ${n_0}_0\ge 3$.  By Markov's inequality, we have
  \begin{equation*}
    \Prb( S_{n_0} \ge \ln {n_0} ) \le \frac{3}{\ln {n_0}}.
  \end{equation*}
\end{proof}


\begin{lemma}\label{lem:S-never-too-big}
  For all $\eps \in\lt]0,1\rt]$ and $b > 1$, if $n_0 \ge \max(10,b^{4/\eps}$ then with $n_1 := b n_0$, whatever Alice does,
  \begin{equation*}
    \Prb\bigl( \exists n\in\{n_0,\dots,n_1-1\}: S_n \ge \ln^2 n \mid B_{n_0}, S_{n_0}\le ln^2 n_0 \bigr) \le \eps.
  \end{equation*}
\end{lemma}
\begin{proof}
  Whatever Alice does, by Lemma~\ref{lem:transition-prbs}, we have $S_{n+1} \le S_n +1$ always and the probability that $S_{n+1} = S_n+1$ is at most $4/n$.  If $S_n \ge \ln^2 n$ for an $n\in\{n_0,\dots,n_1-1\}$, then either $S_{n_0} \ge \ln n_0$, or the number of times $n\in\{n_0,\dots,n_1-1\}$ that $S_{n+1} = S_n +1$ is at least $\ln^2 n - \ln n_0$.  But
  \begin{equation*}
    \ln^2 n - \ln n_0
    \ge
    \ln^2 n_0  - \ln n_0
  \end{equation*}
  The expected number of times $n\in\{n_0,\dots,n_1-1\}$ that $S_{n+1} = S_n +1$ is
  \begin{equation*}
    \sum_{n=n_0}^{n_1-1} \frac{4}{n} \le 4( \ln(n_1/n_0) + 1/n_0 - 1/n_1 ) \le 4\ln(n_1) -\ln(n_0).
  \end{equation*}
  (We have used the well-known bound $\sum_{\ell=m}^n \frac{1}{\ell} \le \ln(n/(m-1))$.)

  Hence, the probability that we have $S_n \ge \ln^2 n$ for some $n\in\{n_0,\dots,n_1-1\}$ is at most
  \begin{equation*}
    \frac{  4\ln(n_1) -\ln(n_0)  }{  \ln^2 n_0  - \ln n_0  }
    =
    \frac{  4\ln(n_1/n_0) -1  }{ \ln n_0 - 1 }
    =
    \frac{  4\ln b -1  }{  \ln n_0 -  1 }
    \le
    \eps.
  \end{equation*}
  The last inequality follows from $4\ln b \le \eps \ln n_0 \le \eps \ln n_0 + 1-\eps$ (as $\eps \le 1$).
\end{proof}


We are now ready to prove Lemma~\ref{lem:T-will-decrease:boost}.

\begin{proof}[of Lemma~\ref{lem:T-will-decrease:boost}]
  For $j=0,1,2,\dots$, denote by $U_j$ the event that
  \begin{equation*}
    \forall n \in \{(2j+1) n_0,\dots,(2j+3)n_0\}\colon \;\; T_{n+1} \ge T_n.
  \end{equation*}
  By Lemma~\ref{lem:T-will-decrease}, we have, for each~$j$,
  \begin{equation*}
    \Prb\Bigl( U_j  \Bigm|  B_{(2j+1)n_0}, T_{(2j+3)n_0} \ge 2, S_{(2j+1)n_0}\le \ln^2((2j+1)n_0) \Bigr) \le 6/7,
  \end{equation*}
  so
  \begin{align}
    &\Prb\Bigl( \forall j=0,\dots, a-1\colon U_j \Bigm| B_{n_0}, T_{n_0} \ge 2, S_{n_0}\le \ln^2(n_0) \Bigr)
    \\
    &= \notag \prod_{j=0}^{a-1}\Prb\Bigl( U_j \Bigm| B_{n_0}, T_{n_0} \ge 2, S_{n_0}\le \ln^2(n_0), \forall i=0,\dots, j-1\colon U_i \Bigr)
    \\
    &\le \tag{$*$}\label{eq:T-will-decrease:boost:average} \prod_{j=0}^{a-1}( \eps + 6/7 )
    \\
    &\le \tag{$**$}\label{eq:T-will-decrease:boost:34pluseps} (7/8)^a
    \\
    &\le \notag e^{\ln(2/\eps)} = \eps/2.
  \end{align}
  The last inequality follows from $10 > 1/\ln(8/7)$; inequality~\eqref{eq:T-will-decrease:boost:34pluseps} follows from $\eps \le 1/56$.

  As for inequality~\eqref{eq:T-will-decrease:boost:average}, we note that $T_{2j n_0}=1$ implies the event $\complement U_{j-1}$ ($\complement$ is the complement), and, by Lemma~\ref{lem:S-never-too-big} the probability that there is exists an $n=n_0,\dots, n_1$, with $S_n > \ln^2 n$ is at most~$\eps$.

  An application of Lemma~\ref{lem:common-edges} gets rid of the conditioning on $S_{n_0}\le \ln^2(n_0)$: Indeed, since from $n_0 \ge e^{6/\eps}$ we have $3/\ln(n_0) \le \eps/2$, this adds another $\eps/2$ to the final probability.  The bound, in total, is $\eps$.
\end{proof}

Note that Lemma~\ref{lem:T-will-decrease:boost} also gets rid of the conditioning on $S_n \le \ln^2 n$.

We are now ready to complete the proof of the main theorem.

\begin{proof}[of Theorem~\ref{thm:main:pedigreegraph}.]
  Let $\eps' \in\lt]0,\nfrac{1}{2}\rt[$ be given.  Set $t := 6/\eps'$.  Since $T_{\sq}$ can only increase when an isolated vertex is created, we have $T_n \le Y$, for all $n\ge 4$, where~$Y$ is the number of isolated vertices.  Hence, by Lemma~\ref{lem:isol-vert} and Markov's inequality, we have
  \begin{equation*}
    \Prb\Bigl( \exists n \ge 4\colon \;\; T_n \ge t+1 \Bigr)
    \le
    \Prb( Y \ge t )
    \le
    \Exp(Y) / t
    = \eps'/3.
  \end{equation*}

  Now take $n_0' \ge 12/\eps' +2$, and large enough to apply Lemma~\ref{lem:T-will-decrease:boost} $n_0:=n_0'$ and to $\eps := \eps'/3t$ (note that this is less than $1/56$).  Denote by $a$ be the number defined in that lemma.  Applying the lemma~$t$ times, for $n_0$ ranging over $n'_0 +j2a n'_0$, $j=0,\dots,t-1$, the probability that we fail at least once to obtain a decrease in the number of connected components, $T_\sq$, is at most $\eps'/3$.  So, with probability at least $1-2\eps'/3$, we must have $T_{n_0}=1$ for one of these $n_0$'s or for an $n$ between $n'_0 +(t-1)2a n'_0$ and $n'_0 +t2a n'_0$.

  Finally, since $n_0' \ge 12/\eps' +2$, by Lemma~\ref{lem:isol-vert}, with probability $1-\eps'/3$, $T_\sq$ will not increase after $n'_0$, and hence, with probability $1-\eps'$, will drop to~1 and stay there for all eternity.  Bob wins.
\end{proof}

\section{Some Open Questions}
There are two questions which we believe should be asked in the context of our result.

Firstly, are there other polytopes whose graphs are not complete, but the minimum degree is asymptotically that of a complete graph?  Could that even be the case for the Traveling Salesman Problem polytope itself?

Secondly, in view of the Traveling Salesman Problem polytope, it would be interesting to find other combinatorial conditions on cycles which are implied by the adjacency of the corresponding vertices on the TSP polytope.  The pedigree graph connectedness condition is derived from an extension of the TSP polytope, but maybe there are other combinatorial conditions without that geometric context.  The graph resulting from such a condition might be ``closer'' to the actual TSP polytope graph.

\section*{Acknowledgments}
The authors would like to thank Kaveh Khoshkhah for pointing us to the idea of analyzing the pair $(S,T)$ of random variables.


\end{document}